\documentclass[11pt]{article}

\usepackage{manyfoot}%
\usepackage{amsmath, amsthm, amssymb,color,booktabs,physics}  %cite
\usepackage{graphicx}
\usepackage[colorlinks]{hyperref} %pagebackref 

\usepackage{cleveref}
\usepackage{tikz}

\newcommand{\nc}{\newcommand}
\nc{\rnc}{\renewcommand}

\newcommand{\proj}[1]{\left|#1\right\rangle\left\langle #1\right|}

\DeclareMathOperator{\Par}{Par}
\DeclareMathOperator{\Per}{Per}

\DeclareMathOperator{\sgn}{sgn}

\DeclareMathOperator{\Wg}{Wg}

\DeclareMathOperator{\QMA}{\mathsf{QMA}}

\def\be#1\ee{\begin{equation}#1\end{equation}}
\def\ba#1\ea{\begin{align}#1\end{align}}
\def\bas#1\eas{\begin{align*}#1\end{align*}}
\def\bpm#1\epm{\begin{pmatrix}#1\end{pmatrix}}
\nc{\non}{\nonumber}
\nc{\nn}{\nonumber}
\nc{\eq}[1]{(\ref{eq:#1})}
\nc{\eqs}[2]{(\ref{eq:#1}) and (\ref{eq:#2})}
\rnc{\L}{\left} 
\nc{\R}{\right}
\nc{\ra}{\rightarrow}
\nc{\ot}{\otimes}
\def\bea#1\eea{\begin{eqnarray}#1\end{eqnarray}}
\def\beas#1\eeas{\begin{eqnarray*}#1\end{eqnarray*}}

\newtheorem{thm}{Theorem}
\newtheorem*{thm*}{Theorem}

\newtheorem{cor}[thm]{Corollary}
\newtheorem{lem}[thm]{Lemma}

\newtheorem{proto}{Protocol}

\theoremstyle{definition}
\newtheorem{remark}{Remark}

\newtheorem{dfn}[thm]{Definition}
\theoremstyle{plain}

\makeatletter
\newtheorem*{rep@theorem}{\rep@title}
\newcommand{\newreptheorem}[2]{%
\newenvironment{rep#1}[1]{%
 \def\rep@title{#2 \ref{##1} (restatement)}%
 \begin{rep@theorem}}%
 {\end{rep@theorem}}}
\makeatother

\newreptheorem{thm}{Theorem}
\newreptheorem{lem}{Lemma}

\nc\eps{\epsilon}

\def\Usch{U_{\text{Sch}}}

\nc\cA{\mathcal{A}}
\nc\cB{\mathcal{B}}
\nc\cC{\mathcal{C}}
\nc\cD{\mathcal{D}}
\nc\cE{\mathcal{E}}
\nc\cF{\mathcal{F}}
\nc\cG{\mathcal{G}}
\nc\cH{\mathcal{H}}
\nc\cI{\mathcal{I}}
\nc\cJ{\mathcal{J}}
\nc\cK{\mathcal{K}}
\nc\cL{\mathcal{L}}
\nc\cM{\mathcal{M}}
\nc\cN{\mathcal{N}}
\nc\cO{\mathcal{O}}
\nc\cP{\mathcal{P}}
\nc\cQ{\mathcal{Q}}
\nc\cR{\mathcal{R}}
\nc\cS{\mathcal{S}}
\nc\cT{\mathcal{T}}
\nc\cU{\mathcal{U}}
\nc\cV{\mathcal{V}}
\nc\cW{\mathcal{W}}
\nc\cX{\mathcal{X}}
\nc\cY{\mathcal{Y}}
\nc\cZ{\mathcal{Z}}

\def\bp{\mathbf{p}}
\def\bq{\mathbf{q}}

\nc\bbC{\mathbb{C}}
\DeclareMathOperator*{\E}{\mathbb{E}}
\DeclareMathOperator*{\bbE}{\mathbb{E}}
\nc\bbF{\mathbb{F}}
\nc\bbM{\mathbb{M}}
\nc\bbN{\mathbb{N}}
\nc\bbR{\mathbb{R}}
\nc\bbZ{\mathbb{Z}}

\nc\benum{\begin{enumerate}}
\nc\eenum{\end{enumerate}}
\nc\bit{\begin{itemize}}
\nc\eit{\end{itemize}}

\newcommand{\secref}[1]{Section~\ref{sec:#1}}
\newcommand{\appref}[1]{Appendix~\ref{sec:#1}}
\newcommand{\lemref}[1]{Lemma~\ref{lem:#1}}
\newcommand{\thmref}[1]{Theorem~\ref{thm:#1}}

\nc{\todo}[1]{\textcolor{red}{todo: #1}}
\nc{\Anote}[1]{\textcolor{red}{Aram note: #1}}

\def\begsub#1#2\endsub{\begin{subequations}\label{eq:#1}\begin{align}#2\end{align}\end{subequations}}
\nc\mnb[1]{\medskip\noindent{\bf #1}}

\nc{\pder}[2]{\frac{\partial {#1}}{\partial {#2}}}
\nc{\p}{\partial}
\nc{\ssym}[2]{\vee^{#2}\bbC^{#1}}
\nc{\psym}[2]{P_{\text{sym}}^{#1,#2}}
\nc{\dimsym}[2]{#1[#2]}
\nc{\dsym}[1]{\dimsym{d}{#1}}

\begin{document}

\title{Approximate orthogonality of permutation operators, with
  application to quantum information}
\author{Aram W. Harrow\footnote{Center for Theoretical Physics,
    Massachusetts Institute of Technology, {\tt aram@mit.edu}}}
\date{\today}
\maketitle 

\begin{abstract}
  Consider the $n!$ different unitary matrices that permute $n$
  $d$-dimensional quantum systems.  If $d\geq n$ then they are
  linearly independent.  This paper discusses a sense in which they
  are approximately orthogonal (with respect to the Hilbert-Schmidt
  inner product, $\langle A,B\rangle = \tr A^\dag B/\tr I$) if
  $d\gg n^2$, or, in a different sense, if $d\gg n$.  Previous work
  had shown pairwise approximate orthogonality of these matrices, but
  here we show a more collective statement, quantified in terms of the
  operator norm distance of the Gram matrix to the identity matrix.
  
  This simple point has several applications in quantum information
  and random matrix theory: (1) showing that random maximally
  entangled states resemble fully random states, (2) showing that
  Boson sampling output probabilities resemble those from Gaussian
  matrices, (3) improving the Eggeling-Werner scheme for multipartite
  data hiding, (4) proving that the product test of Harrow-Montanaro
  cannot be performed using LOCC without a large number of copies of
  the state to be tested, (5) proving that the purity of a quantum
  state also cannot be efficiently tested using LOCC, and (6,
  published separately) helping prove that poly-size random quantum
  circuits are poly-designs.  
\end{abstract}

\section{Introduction}
\subsection{Permutations and quantum states}
When quantum states have symmetries under permutation or collective
rotation, it is possible to reduce the
number of parameters in a problem.  But this may come at a cost in
complexity, for example if the small number of parameters label basis
states in irreducible representations which lack simple constructions.

The main focus of the paper is inspired by the following point:
matrices on $(\bbC^d)^{\ot n}$ that commute with collective unitary
rotations are known to be linear combinations of the permutations of
the $n$ qudits (see below for precise definitions).  If $d\gg n$, then
this indeed reduces the number of parameters from $d^{2n}$ to $n!$.
However, these parameters are coefficients of permutation matrices
that are not quite orthogonal to one another (again, in a sense that
we will clarify below).  We will argue that they are {\em almost}
orthogonal, in a manner that suffices for most applications, when $d
\gg n^2$.

To make these claims more precise, we introduce some definitions.
Denote the symmetric group on $n$ elements by $\cS_n$.  This has a
representation $P_d$ on $(\bbC^d)^{\ot n}$ in which the $n$ qudits are
permuted.  Formally, if $\pi\in\cS_n$, then
\be P_d(\pi)=
\sum_{i_1,\ldots,i_n\in [d]} \ket{i_1,\ldots,i_n}\bra{i_{\pi(1)},\ldots,i_{\pi(n)}},
\label{eq:Pmatrix}\ee
where $[d]:=\{1,\ldots,d\}$.
The definition is chosen so that $P_d(\pi_1)P_d(\pi_2) =
P_d(\pi_1\pi_2)$, that is, $P_d$ is a representation.

Let $M_d$ denote the set of $d\times d$ complex matrices and $\cU_d$
the subset of unitary matrices.
One place where the permutation matrices arise is when considering
operators $A\in M_d^{\ot n}$ that commute with every $X^{\ot n}$ for $X\in\cU_d\}$.    Such $A$ can be written as (see Thm 4.1.13 of 
\cite{GW09} or Cor 4 of \cite{Har-sym})
\be A = \sum_{\pi\in\cS_n} a_\pi P_d(\pi),
\label{eq:Pd-decomp}\ee
for some coefficients $a_\pi\in\bbC$.  This decomposition is useful because it reduces the
number of parameters needed to describe $A$.  However, it is inconvenient that the terms
in \eq{Pd-decomp} are not orthogonal.  We will see this in more detail in our applications below, where
\eq{Pd-decomp} becomes useful precisely when we can establish an approximate orthogonality
for the $P_d(\pi)$ matrices.

We will use the following normalized Hilbert-Schmidt inner product:
\be \langle A,B\rangle := \frac{\tr A^\dag B }{\tr I} = \frac{\tr
  A^\dag B}{d^n} = \bra{\Phi_d}^{\otimes n}(I \otimes A^\dag B )\ket{\Phi_d}^{\otimes n} ,\ee
where $\ket{\Phi_d} = \frac{1}{\sqrt{d}} \sum_{i=1}^d
\ket{i,i}$ is the standard maximally entangled state.

\subsection{Overview of results}
In  \secref{approx-ortho} we will show that
the $n!$ different $P_d(\pi)$ are approximately orthogonal, not only
pairwise (which is rather trivial to show), but in a certain global
sense as well.  Specifically we define the $n!$-dimensional Gram
matrix $G^{(n,d)}$ whose $\pi_1,\pi_2$ entry is the inner product
$\langle P_d(\pi_1),P_d(\pi_2)\rangle$ and argue that it is close to
the identity in operator norm.  Previous work typically focused on
entrywise bounds on $G^{-1}$, and while these often obtained much
sharper results, our applications will rely on the operator norm
estimates presented here.

One easy consequence of this approximate orthogonality relation is to
controlling various norms of linear combinations of permutation
matrices, as we discuss in \secref{norms}.  A less obvious, but still
easy, application is showing that the lower moments of random
bipartite states are close to the moments of random maximally
entangled states, which we describe in \secref{states}.  This is
related to the well-known fact that the entries of Haar-random
unitaries appear to be nearly Gaussian, again when we examine only the
low moments. This in turn has application to improving the parameters
in boson sampling, as we discuss in \secref{bosons}.

The next family of results involves limitations on multi-party quantum
operations where the parties are connected by only classical
communication.  More generally, we consider measurements that remain
valid when the partial transpose operator is applied to a subset of
systems, and that commute with rotations of the form $U^{\otimes n}$.
It turns out that these operators are severely constrained and we use
this to analyze the Eggeling-Werner data hiding scheme and the
complexity of purity testing in \secref{hiding} and establish
limitations on product tests in \secref{prod-purity}.  A further
application has appeared in \cite{BHH-designs}, where this approximate
orthogonality is used to analyze the convergence speed of the
low-order moments of random unitary quantum circuits.

Two appendices explore further topics. \appref{rep-theory} fleshes out
some calculations used in \secref{approx-ortho} and
\Cref{app:partitions} explains how replacing Haar uniform unitaries
with other distributions, such as random classical reversible
operations, does not yield the same structure.

\section{Approximate orthogonality}\label{sec:approx-ortho}

\subsection{Statement of results}\label{sec:approx-ortho-statement}
This section gives a quantitative statement of the approximate orthogonality of permutation operators.

First we relate the inner product between a pair of permutations to a natural metric on
the group of permutations.  Observe that
\be \tr P_d(\pi) = d^{c(\pi)}, \ee
where $c(\pi)$ counts the number of cycles of $\pi$.  Let $T_n\subset S_n$ be the set of
$\binom{n}{2}$ transpositions, and let $\Gamma_n := \Gamma(S_n,T_n)$ be the Cayley graph
of $S_n$ defined by this generating set; i.e. the vertices are $S_n$ and there is an edge
between $\pi_1$ and $\pi_2$ iff $\pi_1^{-1}\pi_2 \in T_n$.  Define $|\pi|$ to be the
minimum number of transpositions necessary to obtain $\pi$ from $e$.  Since graph distance
is invariant under multiplication by $S_n$, $|\cdot|$ satisfies the triangle inequality:
\be |\pi_1\pi_2| \leq |\pi_1| + |\pi_2|,\label{eq:triangle}\ee
Observe also that $|\pi| = n - c(\pi)$.  We now calculate
\be \langle P_d(\pi_1),P_d(\pi_2)\rangle = \frac{\tr P_d(\pi_1^{-1}\pi_2)}{ d^n} =
d^{c(\pi_1^{-1}\pi_2)-n} = d^{-|\pi_1^{-1}\pi_2|}.
\label{eq:pairwise}\ee
Thus, $P_d(\pi_1)$ and $P_d(\pi_2)$ are approximately orthonormal when $d$ is large and/or
when $\pi_1$ and $\pi_2$ are far apart in the transposition metric.

The main goal of this paper is to extend the pairwise approximate orthogonality of
\eq{pairwise} to a certain notion of global approximate orthogonality.  In particular we
will show that the $P_d(\pi)$ are close to an orthonormal basis.  In general, a
collection of vectors with pairwise small inner products does not have to be close to an
orthonormal basis, as we will discuss further in \Cref{app:partitions}.  The fact that
the $P_d(\pi)$ matrices are close to an orthonormal basis will be the key property of them
that we use in most of our applications.

Define the $n!\times
n!$ Gram matrix $G^{(n,d)}$ by
\be G^{(n,d)}_{\pi_1,\pi_2} = \langle P_d(\pi_1),P_d(\pi_2)\rangle %= d^{c(\pi_1^{-1}\pi_2)-n}
= d^{-|\pi_1^{-1}\pi_2|},\ee
Observe that $G^{(n,d)}$ has ones on the
diagonal, and positive powers of $1/d$ in every off-diagonal entry.
Thus we have \be \lim_{d\ra \infty} G^{(n,d)} = I_{n!},\ee corresponding
to the fact that different permutations approach orthogonality as
$d\ra\infty$.

To make this fact useful, we need to know how quickly this limit
converges as a function of $n$.  Naively, we can observe that there
are $n!-1$ off-diagonal terms per row, each $\leq 1/d$, so they add up
to something small if $d \gg n!$.   But much better bounds are
possible.

\pagebreak[2]

\begin{lem}[approximate orthogonality]\label{lem:bound-eig}~\\
\benum
\item  $G^{(n,d)}$ is always positive semidefinite, has trace $n!$, and is invertible
if and only if $n \leq d$.
\item
\be \frac{1}{n!}\|G^{(n,d)} - I_{n!}\|_1 \leq \sqrt{2}\frac{n}{d}.
\label{eq:G-S1-bound}\ee
\item
\begin{subequations}\label{eq:G-extremal}\ba
\lambda_{\min}(G^{(n,d)})  &= \prod_{j=1}^{n-1} \L(1-\frac{j}{d}\R)
\geq 1-\frac{n(n-1)}{2d} \\
\lambda_{\max}(G^{(n,d)})  &= \prod_{j=1}^{n-1} \L(1+\frac{j}{d}\R)
\leq e^{\frac{n(n-1)}{2d}} \label{eq:lambda-max-G}
\ea\end{subequations}
Our applications will mostly rely on the fact the following simplified bounds:
\ba \|G^{(n,d)} - I_{n!}\|_\infty  & \leq \frac{n^2}{d}
\qquad \text{if }n^2\leq d \label{eq:G-bound} \\
\|G^{(n,d)} - I_{n!}\|_{1\ra 1} & \leq e^{\frac{n^2}{2d}}-1 \label{eq:G-1-1},
\ea
where the $1\ra 1$ norm of a matrix means the maximum sum of absolute values of entries of any row.
\eenum
\end{lem}

We see that there are a few different regimes.  If $n>d$, then $G^{(n,d)}$ is singular and
is far from $I_{n!}$.  Because of the qualitative difference between the $n>d$ and
$n \leq d$ regimes, the $n\leq d$ case is referred to as the ``stable range'' in the
context of Schur-Weyl duality.  If $n\leq O(d)$, then the {\em average} eigenvalue of
$G^{(n,d)}$ is close to 1, even though the top and bottom eigenvalues will be
exponentially large and exponentially close to zero respectively.  Finally, if
$n\leq O(\sqrt{d})$, then $G^{(n,d)}$ will be close to $I_{n!}$ in operator norm.

There are two proofs of \lemref{bound-eig}, both requiring some facts from representation
theory.  Using precise statements about the dimensions of irreps of $\cU_d$ and $\cS_n$,
we can calculate the exact formula for the eigenvalues of $G$ and their multiplicities.
We will do this below in \lemref{exact-eig}.  However, part 1 of \lemref{bound-eig} and \cref{eq:lambda-max-G,eq:G-bound}
can also be proved using only a few simple facts about the symmetric and antisymmetric
subspaces.  We give this proof here.

First we recall some facts about the symmetric and antisymmetric subspaces.
Define $\ssym d n$ to be the symmetric subspace of $(\bbC^d)^{\ot n}$,
meaning the set of vectors that is invariant under each $P_d(\pi)$.
We will also use the antisymmetric subspace $\wedge^n \bbC^d$, which
is the set of vectors invariant under each $P_d(\pi)\sgn(\pi)$, where
$\sgn(\pi)$ is defined to be the sign of $\pi$.  The dimensions of
these subspaces are known to be given by $\dim\ssym d n =
\binom{d+n-1}{n} =: \dsym n$ and $\dim\wedge^n \bbC^d =\binom d n$.
For readers unfamiliar with the properties of the symmetric subspace,
Ref.~\cite{Har-sym} gives a review from a quantum-information perspective.

\begin{proof}[Proof of parts 1 and 3 of \lemref{bound-eig}]
For part 1, we observe that $G$ is a Gram matrix, so is automatically
positive semi-definite.  It has dimension $n!$ and ones along its diagonal, so $G$ has trace $n!$.  It is invertible if and only if the matrices
$P_d(\pi)$ are linearly independent.  If $n\leq d$, then the linear
independence of these matrices can be seen by considering their action
on the state $\ket 1 \ot \ket 2 \ot \cdots \ot \ket n \in
(\bbC^d)^{\ot n}$.
To show that $G$ is singular when $n>d$, we define the vector
$\ket\zeta := \sqrt{\frac{d^n}{n!}}\sum_{\pi\in\cS_n} \sgn(\pi)
\ket\pi$.  Now calculate
\bas \bra\zeta G \ket
\zeta
&= \frac{1}{n!}\sum_{\pi_1, \pi_2} d^{c(\pi_1^{-1}\pi_2)} \sgn(\pi_1)\sgn(\pi_2)
\\& = \sum_{\pi\in\cS_n} d^{c(\pi)}\sgn(\pi)
\\& = \sum_{\pi\in\cS_n} \tr P_d(\pi)\sgn(\pi)
\\& = \dim \wedge^n\bbC^d = \binom{d}{n}.\eas
When $n>d$, this expression is 0.  Since $G$ is positive semidefinite,
it follows that it must have an eigenvalue equal to 0.

For part 3, we observe that the sum of the $\pi_1$ row of $G$ is
\begin{subequations}\label{eq:G-row}\ba
\sum_{\pi_2\in \cS_n} G_{\pi_1,\pi_2} & = 
  \sum_{\pi_2\in \cS_n} d^{c(\pi_1^{-1}\pi_2)-n} \\
&= \sum_{\pi\in \cS_n} d^{c(\pi)-n}\\
&= d^{-n}\sum_{\pi\in \cS_n} \tr P_d(\pi)\\
&= \frac{n!}{d^n}\dsym n = \frac{d+n-1!}{d!\cdot d^n} \\
& = \prod_{j=1}^{n-1} \L(1 + \frac{j}{d}\R)
\ea \end{subequations}
Finally, we use the inequality $1+x\leq e^x$ (which holds for all $x$)
to upper-bound the last equation with $e^{\frac{n(n-1)}{2d}}$.  This yields \eq{G-1-1} which implies \eq{G-bound} and in turn \eq{lambda-max-G}.

\end{proof}
\begin{remark}\label{rem:cayley}
An even simpler proof of a nearly equivalent bound was found by Kevin Zatloukal.  The
idea is that $|\cdot|$ describes a metric on a Cayley graph of degree
$\binom{n}{2}$.  Thus, there are at most $\binom{n}{2}^k$ permutations
with $|\pi|=k$, and we have
$$\sum_{\pi\in\cS_n} d^{-|\pi|} \leq
\sum_{k\geq 0} \binom{n}{2}^k d^{-k} = \L(1 - \frac{\binom{n}{2}}{d}\R)^{-1}.$$
\end{remark}

Most of the rest of the paper is devoted to applications of
\eq{G-bound}.  For our applications, we do not need any more precise
information about the distribution of eigenvalues.  However, for
completeness, we will describe the exact spectrum of $G^{(n,d)}$.
The answer turns out to involve the representation theory of the
symmetric and unitary groups.  
\begin{lem}\label{lem:exact-eig}
For each $\lambda\in \Par(n,d)$, $G^{(n,d)}$ has
  $\dim^2\cP_\lambda$ eigenvalues, each equal to
\be
\frac{n! \dim \cQ_\lambda^d}{\dim \cP_\lambda d^n}
=\prod_{(i,j)\in\lambda} \L( 1 + \frac{j-i}{d}\R).
\label{eq:exact-eig}\ee
\end{lem}
Here $\cQ_\lambda^d$ and $\cP_\lambda$ are irreps of $\cU_d$ and $\cS_n$ respectively; see \appref{rep-theory} for full details. From \lemref{exact-eig}, we immediately obtain Parts 1 and 3 of
\lemref{bound-eig}.  Part 2 is nontrivial, but was previously derived
in Lemma 6 of \cite{CHW07}.  That paper also gave asymptotically matching lower bounds on
$\|G-I\|_1$ that we will omit here.

\lemref{exact-eig} has been proven previously~\cite{Novak08,ZJ10,BCHHKM11} but using different techniques.  In \appref{rep-theory} we give a new proof using the terminology of quantum information, and based on the fact that $G$ is a Gram matrix.

\subsection{Related work}\label{sec:related}
As noted above, \lemref{exact-eig} has been previously proved, in several
places~\cite{collins06,Novak08,BCHHKM11,ZJ10}.  Ref.~\cite{BCHHKM11} used the same representation-theoretic
argument, pointing out that it can be applied to any representation of any group, while
\cite{Novak08,ZJ10} used properties of the symmetric group to obtain a simple, nearly
self-contained, calculation.

Define the {\em Weingarten matrix} $\Wg^{(n,d)}$ to be $(G^{(n,d)})^{-1}$.  The Weingarten
matrix was first introduced by Collins and \'Sniady~\cite{collins03,collins06}, although
they used a different normalization convention.  Their goals were to calculate matrix elements of
$\E[(U \ot \bar U)^{\ot n}]$ and to derive asymptotic properties such as freeness for
related families of random matrices.
The relevance of the Weingarten matrix can be seen from Cor 2.4 of Ref.~\cite{collins06},
which gives the following exact expression for this
expectation value:  
\be \bbE_U [U^{\ot n} \ot U^{*,\ot n} ]
=\sum_{\sigma,\tau\in\cS_n} (I \ot P_d(\sigma)) \Phi_d^{\ot n}
(I \ot P_d(\tau))^\dag\Wg^{(n,d)}(\sigma,\tau)  \ket{v_\sigma}\bra{v_\tau}
\label{eq:CS-unitary}.\ee
Following \cite{collins03}, note that $\Wg(\sigma,\tau)$ only depends on $\sigma^{-1}\tau$ so we can also denote this
matrix element by $\Wg(\sigma^{-1}\tau)$ and we can refer to $\Wg(\cdot)$ as the Weingarten function.  See \cite{CMN21} for an accessible recent review, and \cite{collins2022moment} for a discussion of applications.

Several papers have studied the asymptotic behavior of $\Wg$ as $d\ra \infty.$  Ref.~\cite{collins06} (in Cor 2.7) derived its leading order behavior:
\be \Wg^{(n,d)}(\sigma) = \text{Moeb}(\sigma)d^{-n-|\sigma|} +  O_n(d^{-n-|\sigma|-2}),
\label{eq:Wg-asym}\ee
where $\text{Moeb}(\sigma)$ is the M\"obius function. If $\sigma$ has $c_k$ cycles of
length $k$ and $C_k := 2k! / k!(k+1)!$ is the $k^{\text{th}}$ Catalan number then
$\text{Moeb}(\sigma) := \prod_k ((-1)^{k-1}C_k)^{c_{k+1}}$.
Since $\text{Moeb}(e)=1$, \eq{Wg-asym} is another way of saying that in the $d\ra\infty$
limit, $\Wg$  (or equivalently  $G$) approaches the identity matrix.   \eq{Wg-asym} does not address the question of how large $d$ needs to be as a function of $n$ in order for the approximation to be accurate.  Later works addressed this question, culminating in \cite{CM17} which showed that \eq{Wg-asym} is nearly sharp when $d \gg n^{7/4}$.  However, as a pointwise estimate on the entries of a rank-$n!$ matrix, this does not immediately imply bounds on the spectrum.

To get some intuition for \eq{CS-unitary} in the regime where
$\Wg\approx I$, we will compare with the case of Gaussian random matrices.
 Let $G$ be  a random complex $d\times d$ Gaussian
matrix whose entries are i.i.d.\ and satisfy $\bbE[G_{ij}]=0$ and $\bbE[|G_{ij}|^2]=1/d$.
The following formula is known as Wick's theorem (or Isserlis' theorem):
\be \bbE_G [G^{\ot n} \ot {G}^{*,\ot n} ] = \sum_{\pi \in \cS_n}
 (I \ot P_d(\pi)) \Phi_d^{\ot n}
(I \ot P_d(\pi))^\dag
\label{eq:Wick}\ee
This resembles \eq{CS-unitary} but with $\Wg$ replaced by the identity
matrix.  Thus \eq{CS-unitary} and the fac that $\Wg\approx I$ together imply that low
moments of unitary matrices are close to moments of complex Gaussian
matrices.  This can be seen as a generalization of the Poincar\'e-Maxwell-Borel Lemma which states that applying a low-dimensional projector to a uniformly random point in a high-dimensional sphere yields an approximately Gaussian distribution. 
Indeed $\Wg\approx I \approx G$ precisely in the regime
where submatrices of unitary matrices look Gaussian. Similar
observations were made earlier by Novak~\cite{Novak07} and
Matsumoto~\cite{Matsumoto13}.
We explore this point further in
\Cref{sec:states,sec:bosons}.

% \section{Mathematical consequences}\label{sec:math}

\section{Spectra and norms of sums of permutations}\label{sec:norms}
One easy consequence of our main result (\lemref{bound-eig}) is that we can control various norms of sums of permutations.
Suppose that $\eps = n^2/d \leq 1$ and consider some operator $A =
\sum_\pi a_\pi P_d(\pi)$ with $a_\pi\in\bbR$.   We would like to estimate various norms of
$A$.

  The 2-norm is:
\be \frac{\tr A^2}{d^n}  = \langle a, Ga\rangle,\ee
which is $\in[1-\eps/2,e^{\eps/2}]\|a\|_2^2$.  This follows directly from the operator inequalities $(1-\eps/2)I\leq G\leq e^{\eps/2}I$.

To bound the $\infty$-norm of $A$, let $\pi= \arg\max_\pi |a_\pi|$.  Then
\be \frac{|\tr [P_d(\pi)^\dag A]|}{d^n} = \left| \sum_{\sigma\in \cS_n}
  a_{\sigma} G_{\pi,\sigma} \right|  \geq |a_\pi|
   \L(1-\sum_{\sigma\neq\pi}G_{\pi,\sigma} \R)
  \geq \|a\|_\infty(1-\eps),\ee
where the last inequality follows from \eq{G-1-1}.
  Using $\langle A,B\rangle \leq \|A\|_\infty \|B\|_1$ with $B =
P_d(\pi)/d^n$, we obtain
\be \|A\|_\infty \geq (1-\eps/2)\|a\|_\infty .\label{eq:op-norm}\ee
On the other hand, the only obvious upper bound is the trivial $\|A\|_\infty \leq
\|a\|_1$.  This is tight when the $a_\pi$ all have the same sign, or when
the $a_\pi\sgn(\pi)$ do.
 Similarly we obtain
\be \|A\|_1 \geq (1-\eps) d^n \|a\|_\infty.\label{eq:S1-lb} \ee

This method does not seem to yield good bounds on the 1-norm of $A$.
The triangle inequality yields the rather weak bound $\|A\|_1\leq d^n \|a\|_1$ which is
usually improved upon by $\|A\|_1 \leq \sqrt{d^n}\|A\|_2 \leq d^n \sqrt{1+\eps} \|a\|_2$.

\section{Random maximally entangled states}\label{sec:states}

Random pure states are known to be nearly maximally entangled.  In this
section we describe one way to formalize this intuition, by proving
that the low moments of random bipartite states resemble those of
random maximally entangled states.  Let $\ket\psi$ be a random unit
vector in $\bbC^{d^2}$, $U$ a Haar-uniform unitary in $\cU_d$ and
$\ket{\varphi_U} := (U \ot I)\ket{\Phi_d}$.  We also use the
convention that $\psi:=\proj\psi$.  Then (using the
representation-theory notation from \Cref{sec:rep-theory}) we can
calculate:

\ba \bbE_\psi [\psi^{\ot n}] &=
\frac{\sum_{\pi\in\cS_n} P_{d^2}(\pi)}{d^2\cdots (d^2+n-1)}=
\frac{1}{\binom{d^2+n-1}{n}}
\sum_{\lambda\in\Par(n,d)}
\proj{\lambda,\lambda} \ot I_{\cQ_\lambda^d}^{\ot 2} \ot
\Phi_{\cP_\lambda}.\\
\bbE_U [\varphi_U^{\ot n}] & =
\sum_{\lambda\in\Par(n,d)}
\frac{\dim \cP_\lambda}{\dim\cQ_\lambda^d\cdot d^n}
\proj{\lambda,\lambda} \ot I_{\cQ_\lambda^d}^{\ot 2} \ot
\Phi_{\cP_\lambda}.
\ea
The ratio between these coefficients, for a fixed $\lambda$, is the
now-familiar
\be \frac{\dim\cQ_\lambda^d n! d^n}{\dim\cP_\lambda\, d^2\cdots
  (d^2+n-1)}
=\prod_{k=1}^{n-1}\L(1+\frac{k}{d^2}\R)^{-1}\cdot
\prod_{(i,j)\in\lambda} \L( 1 + \frac{j-i}{d}\R).\ee
This is again $\geq 1-n^2/2d$ and $\leq 1+n^2/d$, assuming $n^2\leq
d$.

As an example, suppose that $d=d_Ad_B$ corresponding to a decomposition into two systems $A$ and $B$.  Suppose further that we want to estimate the bipartite entanglement between these systems by examining $\tr[(\tr_A\varphi_U)^2]$.  Then $\bbE[\tr[(\tr_A\varphi_U)^2]]$ can replaced by the corresponding average over $\psi$ at a cost of multiplying by a number in the range $[1-\frac 1d,1+\frac 4d]$.  Similar bounds apply to higher moments of $\tr[(\tr_A\varphi_U)^2]$ or to related quantities.

What if instead of applying a random $U^{\ot n}$ to her half of a
maximally entangled state, Alice applies a random $P_d(\pi)$?  This
results in an ensemble that is in a way dual to the uniform
distribution over maximally entangled states.
Define
$\ket{\psi_\pi} = (P_d(\pi)\ot I_{d^n})\ket{\Phi_d}^{\ot  n}$.
Then
$$\bbE_{\pi\in\cS_n} \psi_\pi=
\sum_{\lambda\in\Par(n,d)}
\frac{\dim \cQ_\lambda}{\dim\cP_\lambda^d\cdot d^n}
\proj{\lambda,\lambda} \ot \Phi_{\cQ_\lambda^d} \ot
I_{\cP_\lambda}^{\ot 2}.
$$
The eigenvalues should look familiar; the $\ket{\psi_\pi}$ are none
other than the $\ket{v_\pi}$ defined in \Cref{sec:rep-theory}, and the
resulting state is $\frac{1}{n!}(K^{(n,d)})^\dag K^{(n,d)}$, which we have argued
is isospectral to $G^{(n,d)}/n!$.

A similar duality between $\cS_n$ and $\cU_d$ twirling was explored by
Mitchison in the context of generalized de Finetti
theorems~\cite{dual-dF}.

\section{Boson sampling anticoncentration}\label{sec:bosons}

Boson Sampling is the process of sending $n$ photons through an array
of beam-splitters that couple $m$ optical modes and then measuring
each mode.  It was introduced as a computational task in \cite{AA13}
and is significant because it appears to not be universal for quantum
computing while remaining hard to simulate classically, assuming some
plausible conjectures.
This gives a plausible route to quantum computational supremacy using current technology; 
see \cite{photon-supremacy} for a recent demonstration.

To understand the output distribution of Boson Sampling, first observe
that the beam-splitters define a unitary $U\in \cU_m$, often taken to
be Haar random.  Suppose the $n\ll m$ photons are input into $n$ modes
corresponding to a set $T\subset [m]$. Then the probability of finding
them into $n$ output modes $S\subset [m]$ is
\be \Pr[S] = |\Per(U_{S,T})|^2, \label{eq:Pr-Per}\ee
where $U_{S,T}$ denotes the
submatrix of $U$ with rows corresponding to $S$ and columns
corresponding to $T$. (There is also a $O(n^2/m)$ probability of
finding two or more photons in the same mode. In this case $S$
becomes a multiset, we interpret $U_{S,T}$ to allow repeated rows, and
the RHS of \Cref{eq:Pr-Per} is divided by $s_1!\ldots s_m!$ where
$s_i$ is the number of photons in mode $i$.  We avoid considering this
case by choosing $n \ll \sqrt m$.)
 Recall that the permanent of a matrix $V$ is
\be \Per(V) = \sum_{\pi\in S_n} \prod_{i=1}^n V_{i,\pi(i)}.\ee

Several steps in the analysis of Boson Sampling are simplified by
approximating the submatrices $V := U_{S,T}$ by a Gaussian matrix
$X$. 
We define $X$ to be an $n\times n$ matrix of i.i.d.~complex Gaussians such that
\be \begin{split}
\bbE[X_{i,j}] & = \bbE[V_{i,j}] = \bbE[X_{i,j}^2] = \bbE[V_{i,j}^2] =0 \\
\bbE[|X_{i,j}|^2] & = \bbE[|V_{i,j}|^2] = \frac 1m
\end{split}\ee
By definition these moments of $V$ and $X$ match, but what about
higher moments?  In \Cref{sec:related} we argued that higher moments
are close as well in the regime where $G\approx I$.  In this section
we will show how this implies that low moments of the permanent are
also close.  Note that the notation in this section is chosen to be
consistent with the boson sampling literature and $(n,m)$ here will
turn out to correspond to $(n,d)$ in the rest of the paper.

\begin{thm}\label{thm:boson}
  If $n^2t^2\leq 2m$ and $V,X$ are defined as above then
\be 1-\frac{n^2t^2}{m} 
\leq \frac{\bbE[|\Per(V)|^{2t}]}{\bbE[|\Per(X)|^{2t}]}
\leq 1 + \frac{n^2t^2}{m}.\label{eq:perm-moment-bound}\ee
\end{thm}

%Let $U\in\cU_m$ be a Haar random unitary and let $V$ denote an $n\times n$ submatrix, say from the upper-left corner (by symmetry, they are all the same).

Section 5.1 of \cite{AA13} establishes a similar but incomparable result,
finding that the distribution of unitary submatrices of size $m^{1/6}$ are close in
variational distance to an i.i.d. Gaussian distribution.
\Cref{thm:boson} by contrast works for submatrices with dimension as large as
$O(m^{1/2})$ but controls only low moments and not the entire
distribution.  However, for some applications, such as the
``anticoncentration'' conjecture, this can be enough.

Nezami~\cite{nezami21} used representation theory to give
formulas for the moments of both $|\Per(X)|^2$ and
$|\Per(V)|^2$.  These can be used to establish bounds similar to
\eq{perm-moment-bound} but slightly stronger\footnote{Specifically
  eqns (9), (13) and (19) from \cite{nezami21}, along with the fact
  that $\rho_\lambda(RCRC)\geq 0$, directly imply that
  \be
1 - \frac{n^2t^2}{2m}\approx \prod_{i=0}^{nt-1}\left(1 + \frac im\right)^{-1}
  \leq \frac{\bbE[|\Per(V)|^{2t}]}{\bbE[|\Per(X)|^{2t}]} \leq \prod_{i=1}^{\min(n,t)}\prod_{j=1}^{\max(n,t)} \left(1+\frac{j-i}m\right)^{-1}
\approx 1 - \frac{nt|n-t|}{m}
\label{eq:nezami}  \ee
This observation is due to Sepehr Nezami.
}.
The contribution of this work then is an independent and somewhat simpler proof of a nearly equivalent result.

\begin{proof}
  
Define
\be \ket{S_n} = \frac{1}{n!} \sum_{\pi \in S_n} \ket{\pi} \qquad \text{where}\qquad
\ket{\pi} = \ket{\pi(1)} \ot \cdots \ot \ket{\pi(n)}\in(\bbC^n)^{\ot n}.\ee
Then $\Per(V) = \bra{S_n} V^{\ot n} \ket{S_n}$.
Since the distribution of $V$ is invariant under $V\mapsto e^{i\phi} V$, $\bbE[\Per(V)]=0$ and similarly for $X$. Thus we will focus instead on
\be \bbE[|\Per(V)|^{2t}] = \bbE[\bra{S_n}^{\ot 2t} V^{\ot nt} \ot {V}^{*,\ot nt} \ket{S_n}^{\ot 2t}].
\label{eq:perm-V2t}\ee
If we interpret $\ket{S_n}$ as being a vector in $(\bbC^m)^{\ot n}$ then we can replace $V$ with $U$ in the RHS of \eq{perm-V2t}, obtaining
\be \bbE[|\Per(V)|^{2t}] = \bbE[\bra{S_n}^{\ot 2t} U^{\ot nt} \ot {U}^{*,\ot nt} \ket{S_n}^{\ot 2t}].
\label{eq:perm-V2t-U}\ee
Now apply \eq{CS-unitary} to evaluate the expectation over $U$ and obtain
\be \bbE[|\Per(V)|^{2t}] = \bra{S_n}^{\ot 2t}
\sum_{\sigma,\tau \in \cS_{nt}}
(I \ot P_m(\sigma)) \Phi_m^{\ot nt}
(I \ot P_m(\tau))^\dag\Wg^{(nt,m)}(\sigma,\tau)
\ket{S_n}^{\ot 2t}.\label{eq:perm-Wg}\ee
By contrast, for the moments of a complex Gaussian,
\be \bbE[|\Per(X)|^{2t}] = \bra{S_n}^{\ot 2t}
\sum_{\pi \in \cS_{nt}}
(I \ot P_m(\pi)) \Phi_m^{\ot nt}
(I \ot P_m(\pi))^\dag
\ket{S_n}^{\ot 2t}.\label{eq:perm-Gaussian}\ee
For $\pi\in\cS_{nt}$ we need to evaluate
\be \alpha_\pi := \bra{\Phi_m}^{\ot nt}(I \ot P_m(\pi))\ket{S_n}^{\ot 2t}
= \frac{1}{\sqrt{m^{nt}}} \bra{S_n}^{\ot t} P_m(\pi)\ket{S_n}^{\ot t}
= \frac{1}{n!^t\sqrt{m^{nt}}} \sum_{\sigma,\sigma' \in \cS_n^t}
1_{\sigma' = \pi\sigma}
\label{eq:perm-overlap}.\ee

Before evaluating $\alpha_\pi$, we can make some observations about the moments of the permanent.  Substituting into \eq{perm-Wg} and \eq{perm-Gaussian} we have
\be \bbE[|\Per(V)|^{2t}] = \bra{\alpha}\Wg^{(nt,m)}\ket{\alpha}
\qand \bbE[|\Per(X)|^{2t}] = \braket{\alpha}{\alpha}.\ee
In the $n^2t^2 \ll m$ regime, our control of the spectrum of $\Wg$
lets us relate these quantities.  Indeed
\be 1-\frac{n^2t^2}{m}
\leq \lambda_{\max}(G^{(nt,m)})^{-1}
\leq \frac{\bbE[|\Per(V)|^{2t}]}{\bbE[|\Per(G)|^{2t}]}
\leq \lambda_{\min}(G^{(nt,m)})^{-1}
\leq 1 + \frac{n^2t^2}{m},\ee
where the outer inequalities are valid when $n^2t^2\leq 2m$.  
\end{proof}

We can also use these formulas to calculate some moments of Gaussian
matrices.  It is not trivial since $\alpha_\pi$ will depend on $\pi$
for $t>1$.  However, the case $t=2$ is relatively quick.  Then
\eq{perm-overlap} will depend on the parameter $\ell = w(\pi) :=
|\pi(\{1,\ldots,n\}) \cap \{n+1, \ldots, 2n\}|$.  Let's fix
$\pi\in\cS_{2n}$, choose $\sigma\in\cS_n^2$ at random and calculate
the probability that $\pi\sigma\in\cS_n\times \cS_n$.  This is
$\binom{n}{\ell}^{-1}$.  Thus we find 
\be \alpha_\pi = \frac{1}{\binom{n}{w(\pi)}m^{nt/2}}.\ee

We also want to calculate $|w^{-1}(\ell)|$. This is given by a hypergeometric distribution.
\be |w^{-1}(\ell)| = \binom{n}{\ell}^4 \ell!^2 (n-\ell)!^2 = n!^2 \binom{n}{\ell}^2.\ee

To apply this to the Gaussian case, we substitute into\eq{perm-Gaussian}.
\be \bbE[|\Per(X)|^{4}] = \sum_{\ell}n!^2 \binom{n}{\ell}^2 \binom{n}{\ell}^{-2} m^{-2n}
=(n+1)\frac{n!^2}{m^{2n}}.\ee
This yields an alternate proof of  Lemma 56 of \cite{AA13}.

\section{Partial transposes of permutation operators}\label{sec:PT-perm}

This section will introduce some mathematical tools that will be
relevant to applications involving multipartite quantum systems and
specifically the proofs in \Cref{sec:hiding,sec:prod-purity}.

A
frequently used tool in understanding locality is the PPT (Positive
Partial Transpose) restriction~\cite{PPT1,PPT2}.  The PPT criteria for
seperability of states and measurements is useful in part because it
has an efficient semidefinite program and these same attributes also
make it more amenable to proofs.  In this section we study the
spectrum of permutation operators with the partial transpose applied
to some of the subsystems.  The goal is to establish lemmas that will
be later used in the applications.

If $\{M, I-M\}$ is a two-outcome measurement that can be implemented
by LOCC (local operations and classical communication~\cite{LOCC}),
then a useful relaxation is to require that $M$ and $I-M$ 
remain positive semi-definite whenever any collection of subsystems is
partially transposed.  We call the measurements satisfying this
condition ``PPT'', meaning that measurement operators are
{\em P}ositive under {\em P}artial {\em
  T}ransposition\footnote{In some cases one might want to constrain only the yes or no operators to being PPT.  In this paper we will always take PPT to mean that all measurement outcomes are PPT.}.
Let $M$ act on $n$ systems, and let $S\subseteq [n]$.  Then we let
$M^{\Gamma_S}$ denote $M$ with the indices in $S$ transposed.  In this
notation, an equivalent characterization of the PPT-BOTH condition is
that
\be 0\preceq M^{\Gamma_S}\preceq I \qquad\forall S\subseteq [n]
\label{eq:PPT-both-cond}\ee

In this section, we discuss partial transposes of the operators $P_d(\pi)$.

The relevance of the PPT constraint is that taking the partial
transpose of part of a permutation matrix can result in the largest
eigenvalue increasing dramatically.  
Thus, if $\Gamma$ denotes the
partial transpose, then requiring that $0\leq
M^\Gamma \leq I$ can be a potent constraint in addition to the usual
$0\leq M\leq I$.

We will consider taking the partial transpose of an arbitrary set
$S\subset [n]$ and will denote this operation $^{\Gamma_S}$.  We also
define $\bar{S} := [n] - S$, so that $(S,\bar{S})$ partition $[n]$.
\begin{lem}\label{lem:PPT}
For any $\pi\in\cS_n$, let $k = |S \cap \pi(\bar{S})|$.  Then
$P_d(\pi)^{\Gamma_S}$ has $d^{n-2k}$ non-zero singular values, each
equal to $d^k$.
\end{lem}
This is a generalization of the well-known fact that $\cF_{1,2}^{\Gamma_2} =
d\Phi$, where $\Phi$ is a projection on the maximally entangled state.
In fact, we can say somewhat more about the structure of
$P_d(\pi)^{\Gamma_S}$ (see \cite{eggeling03, ZKW06}), but \lemref{PPT}
is all we need for our argument.

\begin{proof}[Proof of \lemref{PPT}]
Let $X=(P_d(\pi)^{\Gamma_S})^\dag P_d(\pi)^{\Gamma_S}
= P_d(\pi)^{\Gamma_{\bar{S}}} P_d(\pi)^{\Gamma_S}$.  Then the square of the
singular values of $P_d(\pi)^{\Gamma_S}$ are the eigenvalues of $X$.
To represent tensor products of $n$ systems, we will use a superscript
$^{(i)}$ to indicate that a system should be placed in the
$i^{\text{th}}$ position, so that we can list the systems in an order
that is more convenient.
We now calculate
\bas
X &= \sum_{x_1,\ldots, x_n\in [d]\atop y_1, \ldots, y_n \in [d]}
\bigotimes_{i\in S} \ket{x_i }\braket{x_{\pi(i)}}{y_{\pi(i)}}\bra{y_i}^{(i)}
\;\ot\;
\bigotimes_{i\in\bar{S}} \ket{x_{\pi(i)}} \braket{x_i}{y_i}\bra{y_{\pi(i)}}^{(i)}
  \\ & = \sum_{x_1,\ldots, x_n\in [d]\atop y_1, \ldots, y_n \in [d]}
\L( \prod_{i\in \pi(S)\cup \bar{S}} \delta_{x_i, y_i}\R)
\bigotimes_{i\in S} \ket{x_i }\bra{y_i}^{(i)}
\;\ot\;
\bigotimes_{i\in\bar{S}} \ket{x_{\pi(i)}} \bra{y_{\pi(i)}}^{(i)}
\\ & = \sum_{x_1,\ldots, x_n\in [d]\atop y_1, \ldots, y_n \in [d]}
\L(\prod_{i\in \pi(S)\cup \bar{S}} \delta_{x_i, y_i}\R)
\bigotimes_{i\in S} \ket{x_i }\bra{y_i}^{(i)}
\;\ot\;
\bigotimes_{i\in \pi(\bar{S})} \ket{x_i} \bra{y_i}^{(\pi^{-1}(i))}
\eas
We see that a $\delta_{x_i,y_i}$ appears for all $i$ in $\pi(S)\cup
\bar{S}$, or equivalently, all $i$ not contained in $\pi(\bar{S})\cap S$.
Additionally we see that each $\ket{x_i}\bra{y_i}$ appears zero times
for $i\in\bar{S}\cap \pi(S)$, twice for $i\in \pi(\bar{S})\cap S$ and once
otherwise; i.e.~for $i\in (S\cap \pi(S))\cup (\bar{S}\cap \pi(\bar{S}))$.   (To
justify these arguments, recall that $(S,\bar{S})$ and $(\pi(S),\pi(\bar{S}))$
both partition $[n]$.)

We now consider the partition of $[n]$ into
$\bar{S}\cap \pi(S)$, $(S\cap \pi(S))\cup (\bar{S}\cap \pi(\bar{S}))$ and $S\cap \pi(\bar{S})$
and determine the contributions from each.  Note that since $\pi$ is a
permutation, we have
$|\bar{S}\cap \pi(S)| = |S\cap \pi(\bar{S})| (=k)$,

\bit \item For $i\in\bar{S}\cap \pi(S)$, we have an appearance of
$\delta_{x_i,y_i}$, but not of $\ket{x_i}\bra{y_i}$.  Thus this term
contributes the scalar multiple $d$.
\item For $i\in (S\cap \pi(S))\cup(\bar{S}\cap \pi(\bar{S}))$, we have a $\delta_{x_i,y_i}$
  constraint as well as a $\ket{x_i}\bra{y_i}$ term.  Thus, we have
  one appearance of the $d\times d$ identity operator $I_d =
\sum_{x_i\in [d]} \proj{x_i}$ at position $i$.
\item Finally, for $i\in \pi(\bar{S})\cap S$, there is no
$x_i=y_i$ constraint and the total contribution is
$$\sum_{x_i,y_i \in [d]}
\ket{x_i}\bra{y_i}^{(i)} \ot \ket{x_i}\bra{y_i}^{\pi^{-1}(i)}
 = d \Phi^{(i, \pi^{-1}(i))}.$$
\eit
Together, we conclude that
$$X = d^{2k} \bigotimes_{i \in S\cap \pi(\bar{S})} \Phi^{(i, \pi^{-1}(i))}
 \;\ot\;
\bigotimes_{i \in  (S\cap \pi(S))\cup (\bar{S}\cap \pi^{-1}(\bar{S}))}
I_d^{(i)},$$
which has the claimed eigenvalues.
\end{proof}

Since our bounds are often in terms of $|\pi|$, it is convenient to
express \lemref{PPT} using this quantity.  This is possible because we
often are free to choose $S$ arbitrarily.  In some cases, we will need
to choose a single $S$ that works for multiple permutations.
This too is straightforward but yields a weaker bound.
\begin{lem}\label{lem:max-cut}
For any $\pi_1,\ldots,\pi_k\in\cS_n$ there exists $S\subseteq [n]$
such that
\be \sum_{i=1}^k |\pi_i(S)\cap \bar S| \geq
\frac{1}{4}\sum_{i=1}^k |\pi_i|.\label{eq:max-cut-ineq}\ee
In the special case where we have a single $\pi\in\cS_n$ we can find
$S$ such that
\be |\pi(S)\cap \bar S| \geq \frac{|\pi|}{2}.
\label{eq:max-cut-single}\ee
\end{lem}
\begin{proof}
Suppose $S$ is chosen uniformly at random from the subsets of $[n]$.
For each $i\in [k]$, let $m_i$ denote the number of derangements of
of $\pi_i$, i.e.~the number of $x$ such that $\pi_i(x)
\neq x$.   For such $x$, the probability that $x \in \pi_i(S) \cap \bar S$
is $1/4$.  By linearity of expectation, the expectation of
$|\pi_i(S)\cap \bar S|$ is $m_i/4$.  Now suppose that $\pi_i$ has
$c_1$ 1-cycles, $c_2$ 2-cycles, and so on.  Then since a single cycle
of length $j\geq 2$ has $j$ derangements,
\be m_i =n-c_1 = \sum_{j\geq 2} j c_j
\qand
|\pi_i| = \sum_{j\geq 2} (j-1)c_j.\ee
Together this implies that $m_i \geq |\pi_i|$.  Thus \eq{max-cut-ineq}
holds in expectation, and also therefore holds for at least one choice
of $S$.

For the $k=1$ case we will choose $S$ based on the cycle decomposition
of $\pi$.  For a cycle containing elements $x_1,x_2,\ldots,x_j$ we put
$x_1,x_3,x_5,\ldots$ into $S$.  A cycle of length $j$ then contributes
$\lfloor j/2\rfloor$ to $|\pi(S)\cap S|$.  Since
$\lfloor j/2\rfloor \geq (j-1)/2$ we obtain \eq{max-cut-single}.
\end{proof}

Let $M = \sum_\pi m_\pi P_d(\pi)$ satisfy the PPT condition
\eq{PPT-both-cond} and assume that $n\leq d^{1/2}$.
From the bound $\|M\|\leq 1$ and \cref{eq:op-norm} we have $|m_\pi|
\leq 1/(1-n^2/2d) \leq 1+\frac{n^2}{d}$.  Using the PPT condition and the
stronger condition $n \leq d^{1/4}$ we can show a
much stronger bound when $\pi$ is far from $e$.

\begin{lem}\label{lem:PPT-coeff}
If $M = \sum_\pi m_\pi P_d(\pi)$ satisfies \eq{PPT-both-cond} and
$\frac{n^2}{\sqrt{d}} \leq 1$ then
\be
|m_\pi| \leq \L(1+\frac{n^2}{\sqrt d}\R) d^{-|\pi|/2}  
\ee
\end{lem}

\begin{proof}
Let
\be \pi := \arg\max_\pi |m_\pi| d^{|\pi|/2}
\label{eq:pi-arg-max}.\ee
Use \lemref{PPT} to choose $S$ so that
$|\pi(S)\cap\bar S \geq |\pi|/2$ and thus
\be \|P_d(\pi)^{\Gamma_S}\|_1  = d^{n - |\pi(S)\cap\bar S|} \leq d^{n -
  |\pi|/2}\label{eq:P-gamma-bound}.\ee
Then
\ba
1 & \geq \L|\tr \frac{P_d(\pi)^{\Gamma_S}}{d^{n -
    |\pi|/2}}M^{\Gamma_S}\R|
&\text{from H\"older, \eq{PPT-both-cond} and \eq{P-gamma-bound}}
\\
 & = d^{|\pi|/2}
\L| \L\langle M^{\Gamma_S}, P_d(\pi)^{\Gamma_S} \R\rangle\R| \\
& = d^{|\pi|/2} \L|\L\langle M, P_d(\pi)\R\rangle\R| \\
& \geq d^{|\pi|/2} |m_\pi| -
d^{|\pi|/2} \sum_{\pi' \neq \pi} |m_{\pi'}| G_{\pi,\pi'}\\
& \geq d^{|\pi|/2} |m_\pi| (1 - \sum_{\pi'\neq \pi}
d^{\frac{|\pi|-|\pi'|}{2}} d^{-|\pi^{-1}\pi'|})
& \text{by \eq{pi-arg-max}}\\
& \geq d^{|\pi|/2} |m_\pi| (1 - \sum_{\pi'\neq \pi}
d^{-|\pi^{-1}\pi'|/2})
& \text{by the triangle inequality, \eq{triangle}}\\
& \geq d^{|\pi|/2} |m_\pi| (2 - e^{n^2/2\sqrt{d}})
& \text{by \Cref{eq:G-row}}\\
& \geq  d^{|\pi|/2} |m_\pi| / (1+n^2/\sqrt d) &  \text{ using $n^2\leq \sqrt d$}
.\ea
\end{proof}
  
\section{Multipartite data hiding}\label{sec:hiding}
Let $\rho_0, \rho_1$ be density matrices on $n$ $d$-dimensional
systems that commute with all $U^{\ot n}$.  If $\|\rho_0 - \rho_1\|_1$
is large, then of course, given $\rho_b$ for $b\in \{0,1\}$, there is
some (global) measurement that can estimate $b$ with some
non-negligible bias.  Here we will argue that, on the other hand, LOCC
measurements, or even PPT measurements, cannot learn anything about
$b$.  This data-hiding scheme is due to Eggeling and
Werner\cite{eggeling02} who shows that it was secure when $n$ is fixed
and $d\ra\infty$.
Our contribution is to extend their analysis to the case when $n$ is
up to $O(\sqrt{d})$ by using our approximate orthogonality
relationship.

\begin{thm}\label{thm:hiding}
  Let $\rho_0,\rho_1$ be any density matrices on $(\bbC^d)^{\ot n}$
  that commute with all $U^{\otimes n}$, and let $\{M,I-M\}$ be a PPT
  measurement.  If $n \leq d^{1/4}$ then
  \be |\tr M(\rho_0-\rho_1)| \leq \frac{6n^2}{\sqrt d}.  \ee
\end{thm}

\begin{proof}
Let $M_0,M_1$ be the PPT measurement operators corresponding to guessing $0,1$ respectively.  Let $M=M_0-M_1$.  Then
\be -I \leq M^{\Gamma_S} \leq I
\label{eq:M-gamma-bound}\ee
 for any $S\subseteq [n]$, where $\Gamma_S$ corresponds to taking the partial transpose of indices $S$.
Let $\Delta = \rho_0-\rho_1$, so that the bias achieved by the measurement is $\tr M \Delta$.  Observe that $\tr\Delta=0$ and that $[\Delta,U^{\ot n}]=0$ for all $U$.

We can assume WLOG that $[M, U^{\ot n}]=0$ as well.  This is because
$$\tr M\Delta = \bbE_U \tr M U^{\ot n}\Delta (U^\dag)^{\ot n}
= \bbE_U \tr ( (U^\dag)^{\ot n}M U^{\ot n})\Delta.$$
Thus, we can write $M = \sum_{\pi\in\cS_n} m_\pi P_d(\pi)$.

The bias is now bounded by
\ba \tr M\Delta
 &= \sum_{\pi \neq e} m_\pi \tr[ P_d(\pi)\Delta] &\text{because $\tr\Delta=0$} \\
 &\leq \sum_{\pi \neq e} |m_\pi| \|P_d(\pi)\|_\infty \|\Delta\|_1
 &\text{triangle inequality and H\"older} \\
 &= \sum_{\pi \neq e} |m_\pi| \\
 &\leq 2 \sum_{\pi \neq e} d^{-|\pi|/2}  & \text{using \lemref{PPT-coeff}}\\
 & \leq 3 (e^{n^2 / \sqrt{d}}-1) \\
 & \leq 6 n^2/ \sqrt{d}  &  \text{ using $n^2\leq \sqrt d$}.
 \ea
\end{proof}

\thmref{hiding} applies to any two states $\rho_0,\rho_1$ satisfying the symmetry condition, although it is only interesting when $\|\rho_0 - \rho_1\|_1$ is large.  Coming up with one such pair is straightforward, but how many can be constructed simultaneously?  Here we can use Schur duality (c.f.~\Cref{sec:rep-theory}) to show that any state commuting with all $U^{\otimes n}$ must be of the form 
\be \sum_{\lambda\in \Par(n,d)} p_\lambda \rho_\lambda \otimes \tau_{\cQ_\lambda^d},\ee
where $\tau_{\cQ_\lambda^d} := I_{\cQ_\lambda^d} / \dim \cQ_\lambda^d$.
This permits $N = \sum_{\lambda \in \Par(n,d)}\dim P_\lambda$ perfectly orthogonal states.  Since $d\geq n$, $\Par(n,d)$ includes all partitions of $n$, and thus $\sum_{\lambda \in \Par(n,d)}\dim P_\lambda^2 = n!$.  As a result, $N\geq \sqrt{n!}$.  On the other hand, $\Par(n,n) \leq e^{2c\sqrt{n}}$ for $c\approx 1.28$, so $N\leq \sqrt{n!}e^{c\sqrt{n}}$.  This analysis also implies that $\frac 12 \log(n!)$ qubits can be hidden in such states.  If we are content with pairwise approximate distinguishability then exponentially more states can be hidden~\cite{Winter-finger}.

Another application concerns the distinguishability of $n$ copies of the same random state from $n$ copies of independently random states.  As density matrices, these correspond to $\E[\psi^{\otimes n}]$ and $(I/d)^{\otimes n}$ respectively.   If collective measurements are allowed then projecting onto the symmetric subspace will almost perfectly distinguish these states. But the situation is different with LOCC measurement.

\begin{cor}[Local purity tests]
If a PPT measurement is used to distinguish $\E[\psi^{\otimes n}]$ from $(I/d)^{\otimes n}$ it will achieve bias $\leq O(n^2/\sqrt d)$.
\end{cor}

Recently and independently
of this work, sharper upper
and lower bounds were found by Chen, Cotler, Huang and
Li~\cite{learning-memory} who showed that $n=\Theta(\sqrt d)$ copies
are necessary and sufficient for local purity testing.

\section{Limitations of local product tests}\label{sec:prod-purity}

Suppose we are given $n$ copies of a $k$-partite pure state $\ket\psi \in (\bbC^d)^{\otimes k}$.  We would like to know if $\ket\psi$ is close to being a product state $\ket{\psi_1}\ot \cdots \ot\ket{\psi_k}$ or far from any such state.  A natural test for this is to project all $n$ copies of each of the $k$ subsystems onto the $n$-fold symmetric subspace $\ssym d n$.  If all the projections succeed, output ``product'', otherwise output ``not product''.  This test was proposed by \cite{mintert05} and analyzed by \cite{HM13}.
The test can be easily shown to be optimal among a reasonable class of such product tests (see Section 5 of \cite{HM13}), but the projections require entangling operations across the $n$ copies.

How effective can be make product tests without such entangling
operations?  If an LOCC test existed, then it would imply that
$\QMA=\QMA(2)$~\cite{BCY}, and, depending on the accuracy of the test,
this might falsify the Unique Games Conjecture~\cite{BHKSZ12} or the
Exponential Time Hypothesis~\cite{HM13}.  In \cite{HM13} it was proved
that such a test cannot exist for $n=2$.  Here we show it cannot exist
even for larger values of $n$, and even in the easiest case where
$k=2$.

To be more precise we say that a product test consists of a two-outcome measurement $\{M,I-M\}$, corresponding to outcomes ``product'' and ``not product.''  The completeness $c$ is $\min \tr [M\psi^{\ot n}]$ over all product states $\psi$ while the soundness $s$ is $\max \tr[M\psi^{\ot n}]$ over all states $\psi$ with overlap $\leq 1/2$ with any product state. (The constant $1/2$ is arbitrary, however note that no state is orthogonal to all product states.) Define the {\em bias} to be
$b=c-s$.    The standard product test from \cite{mintert05} was proved in \cite{HM13} to have bias $\geq
\Omega(1)$ with $n=2$ and $k$ arbitrary.  However, we will see that this cannot be achieved by a PPT test unless $n$ grows with $d$.

\begin{thm}\label{thm:no-PPT}
If $\{M,I-M\}$ is a PPT product test for $k=2$ acting on $n$ copies of a state,
then its bias $b$ is $\leq O(n^2 / d^{1/4})$.
\end{thm}

Our relation between $n$ and $d$ is tight up to polynomial factors,
since when $n\gg d^2$ then state tomography can be carried out even
with no communication between subsystems.

\begin{proof}[Proof of \thmref{no-PPT}]
  Assume that $n\leq d^{1/8}$ since otherwise the theorem holds trivially.
  Let
  \be \Delta = \bbE_{\ket{\psi_A},\ket{\psi_B}\in \bbC^d}
  [\psi_A^{\ot n} \ot \psi_B^{\ot n}] - \bbE_{\ket{\psi} \in
    \bbC^{d^2}} [\psi^{\ot n}].\ee
  Our goal is to show that $\tr M\Delta$ is small for any PPT measurement $\{M,I-M\}$.

First, we observe that $\Delta$ commutes with $U^{\ot n} \ot V^{\ot
  n}$, and so without loss of generality we can assume that $M$ does as well.  Thus,
the arguments leading to \eq{Pd-decomp} imply that
\be M = \sum_{\pi_A, \pi_B \in \cS_n} m_{\pi_A, \pi_B} P_d(\pi_A) \ot
P_d(\pi_B)
\label{eq:M-decomp}.\ee
For convenience, we will refer to the pair $(\pi_A, \pi_B)$ as a
single permutation $\pi\in\cS_{2n}$.  Formally, we can embed $\cS_n
\times \cS_n$ into $\cS_{2n}$ as the set of permutations that does not
mix $\{1,\ldots,n\}$ and $\{n+1,\ldots,2n\}$.

We will need to develop a variant of \lemref{PPT-coeff} to show that 
\be  m_{\pi} \leq 2 d^{-|\pi|/4}.\label{eq:m-bipartite-bound}\ee
This will imply our desired result as follows:
\ba \tr M\Delta & = \sum_{\pi_A,\pi_B\in \cS_n} m_{\pi_A,\pi_B} \tr
(P_d(\pi_A) \ot P_d(\pi_B))\Delta
\\
& = \sum_{(\pi_A,\pi_B) \neq (e,e)}m_{\pi_A,\pi_B} \tr
(P_d(\pi_A) \ot P_d(\pi_B))\Delta
& \text{since $\tr \Delta=0$} \\
& \leq \sum_{(\pi_A,\pi_B) \neq (e,e)} |m_{\pi_A,\pi_B}|
\\& \leq 2
\sum_{(\pi_A,\pi_B) \neq (e,e)} d^{-\frac{|\pi_A|+ |\pi_B|}{4}}
\\& \leq 2
\L(\L(\sum_{\pi\in\cS_n} d^{-\frac{|\pi|}{4}}\R)^2 -1\R)
\\&\leq 2((e^{n^2/2d^{1/4}})^2-1)
\\&\leq \frac{4n^2}{d^{1/4}}\ea

Now we return to the proof of \eq{m-bipartite-bound}, which
essentially repeats the proof of  \lemref{PPT-coeff} but uses the
multiple-permutation version of \lemref{max-cut}.  The new feature of
this setting is that the
locality constraint here is between $A_1B_1:A_2B_2:\cdots:A_nB_n$
while the permutations $\pi_A$ and $\pi_B$ act on $A_1\ldots A_n$ and
$B_1\ldots B_n$ respectively.  Thus our PPT condition is that
$\|M^{\Gamma_S}\| \leq 1$ where $\Gamma_S$ is a shorthand for the transpose
of systems $\bigcup_{i\in S}\{A_i, B_i\}$.

Following the proof of \lemref{PPT-coeff}, let
\be \pi  := \arg\max_{\pi = \pi_A \times\pi_B} |m_\pi| d^{|\pi|/4}\ee
and use \lemref{max-cut} to find $S\subseteq [n]$ such that
\be
|S\cap \pi_A(\bar{S})| + |S\cap \pi_B(\bar{S})|
\geq \frac{|\pi|}{4} = \frac{|\pi_A| + |\pi_B|}{4} \ee
The rest of the proof is almost identical.
\ba \|P_d(\pi)^{\Gamma_S}\|_1  &= d^{n -
|S\cap \pi_A(\bar{S})|} \cdot d^{n - |S\cap \pi_B(\bar{S})|}
\leq d^{2n -|\pi|/4}\label{eq:P-gamma-2}.\\
1 & \geq
\L|\tr \frac{P_d(\pi)^{\Gamma_S}}{d^{2n -|\pi|/4}}M^{\Gamma_S}\R| \\
 & = d^{|\pi|/4}
\L| \L\langle M, P_d(\pi) \R\rangle\R| \\
& \geq d^{|\pi|/4} |m_\pi| -
d^{|\pi|/4} \sum_{\substack{\pi' \in \cS_n \times \cS_n \\ \pi' \neq \pi}} |m_{\pi'}| G_{\pi,\pi'}\\
& \geq d^{|\pi|/4} |m_\pi| (1 - \sum_{\substack{\pi' \in \cS_n \times \cS_n \\ \pi' \neq \pi}}
d^{\frac{|\pi|-|\pi'|}{4}} d^{-|\pi^{-1}\pi'|})\\
& \geq d^{|\pi|/4} |m_\pi| (1 - \sum_{\substack{\pi' \in \cS_n \times \cS_n \\ \pi' \neq \pi}}
d^{-\frac 34|\pi^{-1}\pi'|})
& \text{by the triangle inequality, \eq{triangle}}\\
& \geq d^{|\pi|/4} |m_\pi|\L (2 - \L(e^{\frac{n^2}{2d^{3/4}}}\R)^2\R)
& \text{by \Cref{eq:G-1-1}} \\
& \geq \frac 12 d^{|\pi|/4} |m_\pi|
\ea

\end{proof}

\appendix

\section{Full spectrum of the Gram matrix}\label{sec:rep-theory}

In this appendix we give a self-contained proof of \lemref{exact-eig}.  The idea is to decompose the permutation action $P_d$ into irreps of $\cS_n$.  We begin with some terminology from representation theory.

Let $\Par(n,d)$ denote the set of partitions of $n$ into $d$ parts;
that is $\lambda\in\Par(n,d)$ if
$\lambda = (\lambda_1,\ldots,\lambda_d)\in \bbZ^d_+$ with
$\lambda_1 \geq \cdots\geq \lambda_d\geq 0$ and
$\sum_{i=1}^d \lambda_i$.  We also identify $\lambda$ with the set of
$(i,j)\in\bbN^2$ with $j\leq \lambda_i$.  Schur duality states that
\be (\bbC^d)^{\ot n} \cong \bigoplus_{\lambda\in\Par(n,d)}
\cQ_\lambda^d \ot \cP_\lambda,\label{eq:schur-iso}\ee where
$\cQ_\lambda^d$ labels an irrep of $\cU_d$ and $\cP_\lambda$ labels an
irrep of $\cS_n$.  Let $\bq_\lambda^d(U)$ and $\bp_\lambda(\pi)$
denote the corresponding group actions of $\cU_d$ and $\cS_n$.  Assume
for convenience that $\bp_\lambda(\pi)$ is always a real orthogonal
matrix.  We let $\Usch$ denote the unitary isomorphism mapping the LHS
of \eq{schur-iso} to the RHS; however, we generally abuse notation and
omit writing $\Usch$.

We will need to make use of the following formulas for the dimensions
of these irreps.  Define $\tilde{\lambda} := \lambda +
(d-1,d-2,\ldots,1,0)$.  Then~\cite{GW09,Stanley-EC2}
\ba \dim \cQ_\lambda^d &=\frac{\prod_{1\leq i<j\leq d}
(\tilde{\lambda}_i - \tilde{\lambda}_j)}{\prod_{m=1}^{d-1} m!}
\label{eq:cQ-dim}\\
\dim \cP_\lambda&= \frac{n!}{\tilde{\lambda}_1! \tilde{\lambda}_2!
\cdots \tilde{\lambda}_d!}
\prod_{1\leq i<j\leq d}
(\tilde{\lambda}_i - \tilde{\lambda}_j)
\label{eq:cP-dim}\ea
We will need the ratio of these dimensions.   One can directly calculate (and see also \cite{Sta71})
\be
\frac{n!\dim \cQ_\lambda^d}{\dim \cP_\lambda}
 = \prod_{i=1}^d \frac{\lambda_i+d-i!}{d-i!}
= \prod_{i=1}^d \prod_{j=1}^{\lambda_i} d-i+j\ee
This last double product can be abbreviated as the product over
$(i,j)\in \lambda$, where $\lambda$ is overloaded to mean both the
partition $\lambda_1,\ldots,\lambda_d$ and the set $\{(i,j): 1\leq j\leq
\lambda_i\}$.

\begin{proof}[Proof of \lemref{exact-eig}]

Let $\{\ket{\pi}:
\pi\in\cS_n\}$ denote a set of orthonormal vectors indexed by the
permutations  and define $\ket{v_\pi} = (I \ot P_d(\pi))\ket{\Phi_d}^{\ot
  n}$.
We also define the maximally entangled states
$\ket{\Phi_{\cP_\lambda}}\in \cP_\lambda \ot \cP_\lambda$ and
$\ket{\Phi_{\cQ_\lambda^d}}\in \cQ_\lambda^d \ot (\cQ_\lambda^d)^*$ to
be unit vectors that are invariant respectively under $\bp_\lambda(\pi) \ot
\bp_\lambda(\pi)$ for all $\pi\in\cS_n$ and
$\bq_\lambda^d(U)\ot\bq_\lambda^d(U)^*$ for all $U\in\cU_d$. (We can omit the $*$  for $\cP_\lambda$ because we have taken
$\bp_\lambda(\pi)$ to be real orthogonal matrices.)  By Schur's Lemma,
these conditions specify $\ket{\Phi_{\cP_\lambda}}$ and
$\ket{\Phi_{\cQ_\lambda^d}}$ uniquely, up to a phase.  To set this
phase, let $\ket{\Phi_d}^{\ot
  n}:=\sum_{\lambda\in\Par(n,d)}\sqrt{\frac{\dim\cQ_\lambda^d\dim\cP_\lambda}{d^n}}
\ket{\lambda,\lambda}\ket{\Phi_{\cQ_\lambda^d}}\ket{\Phi_{\cP_\lambda}}
$.  Thus
\be \ket{v_\pi} =
 \sum_{\lambda\in\Par(n,d)}\sqrt{\frac{\dim\cQ_\lambda^d\dim\cP_\lambda}{d^n}}
\ket{\lambda,\lambda}\ket{\Phi_{\cQ_\lambda^d}}(I \ot
\bp_\lambda(\pi))\ket{\Phi_{\cP_\lambda}}
.\ee

Observe that $\braket{v_{\pi_1}}{v_{\pi_2}} =
\langle P_d(\pi_1), P_d(\pi_2)\rangle$.  Define the matrix $K^{(n,d)} :=
  \sum_{\pi\in\cS_n} \ket\pi\bra{v_\pi}$, and observe that $G^{(n,d)} =
  K^{(n,d)}(K^{(n,d)})^\dag$.  Thus $G^{(n,d)}$ is isospectral to
\ba  (K^{(n,d)})^\dag K^{(n,d)} & =
\sum_{\pi\in\cS_n} \proj{v_\pi} \\
&=n! \sum_{\lambda\in\Par(n,d)}\frac{\dim\cQ_\lambda^d\dim\cP_\lambda}{d^n}
\proj{\lambda,\lambda}\ot \proj{\Phi_{\cQ_\lambda^d}}
\ot \frac{I_{\cP_\lambda}}{\dim\cP_\lambda}
\ot \frac{I_{\cP_\lambda}}{\dim\cP_\lambda}
\ea
\end{proof}

\section{Partitions are not approximately
  orthogonal}\label{app:partitions}

Most conclusions in this paper do not depend strongly on the
properties of $\cS_n$ or $\cU_d$.  As noted in \Cref{rem:cayley}, to
show that $\|G- I\|\leq n^2/2d$, we need only that $G_{x,y} =
d^{-\text{dist}(x,y)}$ where $\text{dist}(\cdot,\cdot)$ is the graph
distance on a graph of degree $\leq n^2/2$.  Could we replace $\cS_n$
with other sets?

Of course for general $N$-dimensional vectors, one can have
$\exp(O(N\eps^2)))$ vectors with pairwise inner product at most
$\eps$, but they must be collectively far from an orthonormal
basis. So pairwise distance certainly does not guarantee any kind of
approximate orthogonality in the collective sense we have discussed.

There is one natural analogue of $\cS_n$ where approximate
orthogonality also turns out to fail.  This example is due to Kevin
Zatloukal. Let $\cP_n$ be the set of partitions of the set
$[n]$.  For example, $\cP_3$ consists of five partitions:
$\{\{1\},\{2\},\{3\}\}$,
$\{\{1,3\},\{2\}\}$,
$\{\{1\},\{2,3\}\}$,
$\{\{1, 2\},\{3\}\}$, and 
$\{\{1, 2, 3\}\}$.   Given a partition $\Pi$, define $[d]^\Pi$ to be the
set of strings $x_1,\dots,x_n \in [d]^n$ where $x_i=x_j$ whenever
$i,j$ are in the same block of $\Pi$.  The corresponding quantum state
is
\be \ket{E_\Pi} := d^{-\frac{\text{number of blocks of }\Pi}{2}} \sum_{x\in
  [d]^\Pi} \ket{x}
.\ee
These states were used in 0811.2597.

Let $G[\cP_n]$ denote the Gram matrix of $\{\ket{E_\Pi}\}$ states,
while we use $G[\cS_n]$ to denote the Gram matrix studied in the rest
of the paper.  Concretely $G[\cP_n]_{\Pi_1,\Pi_2} =
|\braket{\Pi_1}{\Pi_2}|^2$. In both cases we have $1$ on the diagonal
and positive powers of $1/d$ for each off-diagonal entry.  In both
cases, the dimension is exponential in $n$. (The number of partitions
is given by the Bell numbers, which are $\leq n^n$.) However the
interpretation in terms of distances in a low-degree graph does not
exist.  Indeed, if $\Pi_0 = \{\{1,2,\ldots,n\}\}$ and $\Pi_S = \{S,
[n]-S\}$ for some nonempty $S\subset [n]$, then
$|\braket{\Pi_0}{\Pi_S}|^2 = 1/d$ and there are $2^n-2$ choices of
$S$.  As a result the norm of $G[\cP_n]$ is large unless $d \gg 2^n$.

\section*{Acknowledgments}
Thanks to Ashley Montanaro for many helpful discussions especially
about the product test and boson sampling; to Kevin Zatloukal for his
observations in \Cref{lem:bound-eig,app:partitions}; to Fernando
Brand\~ao for discussions about applications to k-designs; to Sepehr
Nezami for discussions surrounding \eq{nezami}; to Beno\^it Collins
and Jon Novak for helping me understand the math literature on this
topic.  \Cref{sec:bosons} benefited from helpful discussions with
Scott Aaronson, Raul Garcia-Patron and Dominik Hangleiter.

Funding is from NSF grants CCF-1452616, CCF-1729369, PHY-1818914 and
the NSF QLCI program through grant number OMA-2016245 as well as NTT
(Grant AGMT DTD 9/24/20).

%% BioMed_Central_Bib_Style_v1.01

\end{document}